\documentclass[letterpaper]{article} % DO NOT CHANGE THIS
\usepackage{aaai24}  % DO NOT CHANGE THIS
\usepackage{times}  % DO NOT CHANGE THIS
\usepackage{helvet}  % DO NOT CHANGE THIS
\usepackage{courier}  % DO NOT CHANGE THIS
\usepackage[hyphens]{url}  % DO NOT CHANGE THIS
\usepackage{graphicx} % DO NOT CHANGE THIS
\usepackage{natbib}  % DO NOT CHANGE THIS AND DO NOT ADD ANY OPTIONS TO IT
\usepackage{caption} % DO NOT CHANGE THIS AND DO NOT ADD ANY OPTIONS TO IT
\title{Keeping Up With the Winner!
\\Targeted Advertisement to Communities in Social Networks}
\author {
    Shailaja Mallick\equalcontrib \textsuperscript{\rm 1},
    Vishwaraj Doshi\equalcontrib \textsuperscript{\rm 2},
    Do Young Eun\textsuperscript{\rm 1}
}
\affiliations {
    \textsuperscript{\rm 1}North Carolina State University\\
    \textsuperscript{\rm 2}Data Science \& Advanced Analytics IQVIA Inc.\\
    mallick.shailaja@gmail.com,
    vishwaraj.doshi@iqvia, 
    dyeun@ncsu.edu
    
}
\usepackage{dblfloatfix}
\usepackage{textcomp}
\usepackage{amsmath}

\usepackage{xcolor}
\usepackage[Symbolsmallscale]{upgreek}
\usepackage{enumerate}
\usepackage{comment}
\usepackage[caption=false]{subfig}
\usepackage{etoolbox}
\usepackage[linesnumbered,ruled,vlined]{algorithm2e}
\usepackage{multicol}
\usepackage{amssymb}
\usepackage{amsthm}

\makeatletter
\newcommand{\onetagleft}{\tagsleft@true}
\makeatother

\makeatletter
\newcommand{\onetagright}{\tagsleft@false}
\makeatother

\newcommand{\R}{\mathbb{R}}

\newcommand{\ones}{\mathbf{1}}

\newcommand{\vecnu}{\boldsymbol \nu}
\newcommand{\vecalpha}{\boldsymbol \alpha}

\newtheorem{theorem}{Theorem}[section]

\newtheorem{proposition}[theorem]{Proposition}
\newtheorem{lemma}[theorem]{Lemma}

\newcommand{\mM}{\mathbf{M}}
\newcommand{\mH}{\mathbf{H}}

\newcommand{\mS}{\mathbf{S}}
\newcommand{\mA}{\mathbf{A}}
\newcommand{\mB}{\mathbf{B}}

\newcommand{\vp}{\mathbf{p}}
\newcommand{\0}{\mathbf{0}}

\newcommand{\vx}{\mathbf{x}}
\newcommand{\vy}{\mathbf{y}}

\newcommand{\vz}{\mathbf{z}}

\newcommand{\vq}{\mathbf{q}}

\newcommand{\vecv}{\mathbf{v}}
\newcommand{\vw}{\mathbf{w}}
\newcommand{\vu}{\mathbf{u}}

\newcommand{\cG}{\mathcal{G}}
\newcommand{\cN}{\mathcal{N}}

\newcommand{\cF}{\mathcal{F}}

\begin{document}

\maketitle

\begin{abstract}
When a new product enters a market already dominated by an existing product, will it survive along with this dominant product? Most of the existing works have shown the coexistence of two competing products spreading/being adopted on overlaid graphs with same set of users. However, when it comes to the survival of a weaker product on the same graph, it has been established that the stronger one dominates the market and wipes out the other. This paper makes a step towards narrowing this gap so that a new/weaker product can also survive along with its competitor with a positive market share. Specifically, we identify a locally optimal set of users to induce a community that is targeted with advertisement by the product launching company under a given budget constraint. To this end, we model the system as competing Susceptible-Infected-Susceptible (SIS) epidemics and employ perturbation techniques to quantify and attain a positive market share in a cost-efficient manner. Our extensive simulation results with real-world graph dataset show that with our choice of target users, a new product can establish itself with positive market share, which otherwise would be dominated and eventually wiped out of the competitive market under the same budget constraint.
\end{abstract}

%----------------------INTRODUCTION------------------------------

\section{Introduction} \label{Introduction}

\textbf{Motivation:} In this competitive world of rapid technological advancements, one of the key strategic decisions that a company launching a new product has to make is how to market itself \cite{cooper1987new, montoya1994determinants}. It can carve out a niche by differentiating itself from its competitors, either by its product design or brand image \cite{fuchs2010evaluating, jalkala2014brand}, and highlighting them through marketing techniques such as advertisements. In such situations, advertisements serve as a medium of communication and play a key role for the firm to achieve this differentiation. There has been evidence that products with no mass advertisements have gained popularity over time as the companies have focused on targeting advertisements to niche communities \cite{dalgic1994niche}. The effectiveness of targeting a small portion of users (communities) who can further spread and influence their neighborhood network via the word-of-mouth, rather than focusing on the whole population, has proven beneficial and also recognized by businesses in the market \cite{armstrong2003marketing, yang2006mining}. 

For example, Nintendo dominated the video game consoles market in the US for almost a decade until the introduction of the PlayStation (PS) \cite{gallagher2002innovation}. Over the next few years, PS was able to establish and eventually dominate the video game console market, forcing Nintendo to focus on hand-held consoles. Subsequently, Xbox initially entered a PS-dominated video game consoles market, and was able to sustain itself in the market \cite{xboxps}. Similarly, in the aviation industry, Airbus entered and managed to sustain in a Boeing dominated market \cite{irwin2004airbus}. Figure \ref{fig:casestudy} shows how before the introduction of Xbox 360, PS 2 already dominated the US market and since the emergence of Xbox, it took some time to mark its presence in the market against its competitor (similar for iPhone 4S and Galaxy S4)\footnote{Real-data collected from Google Trends for the United States region which provides ``insights into search pattern''(\url{https://trends.google.com/trends/?geo=US}).}. 

%---------------------Images of Case Study------------------
\begin{figure}[!t]
    \centering   
    \includegraphics[scale = 0.17]{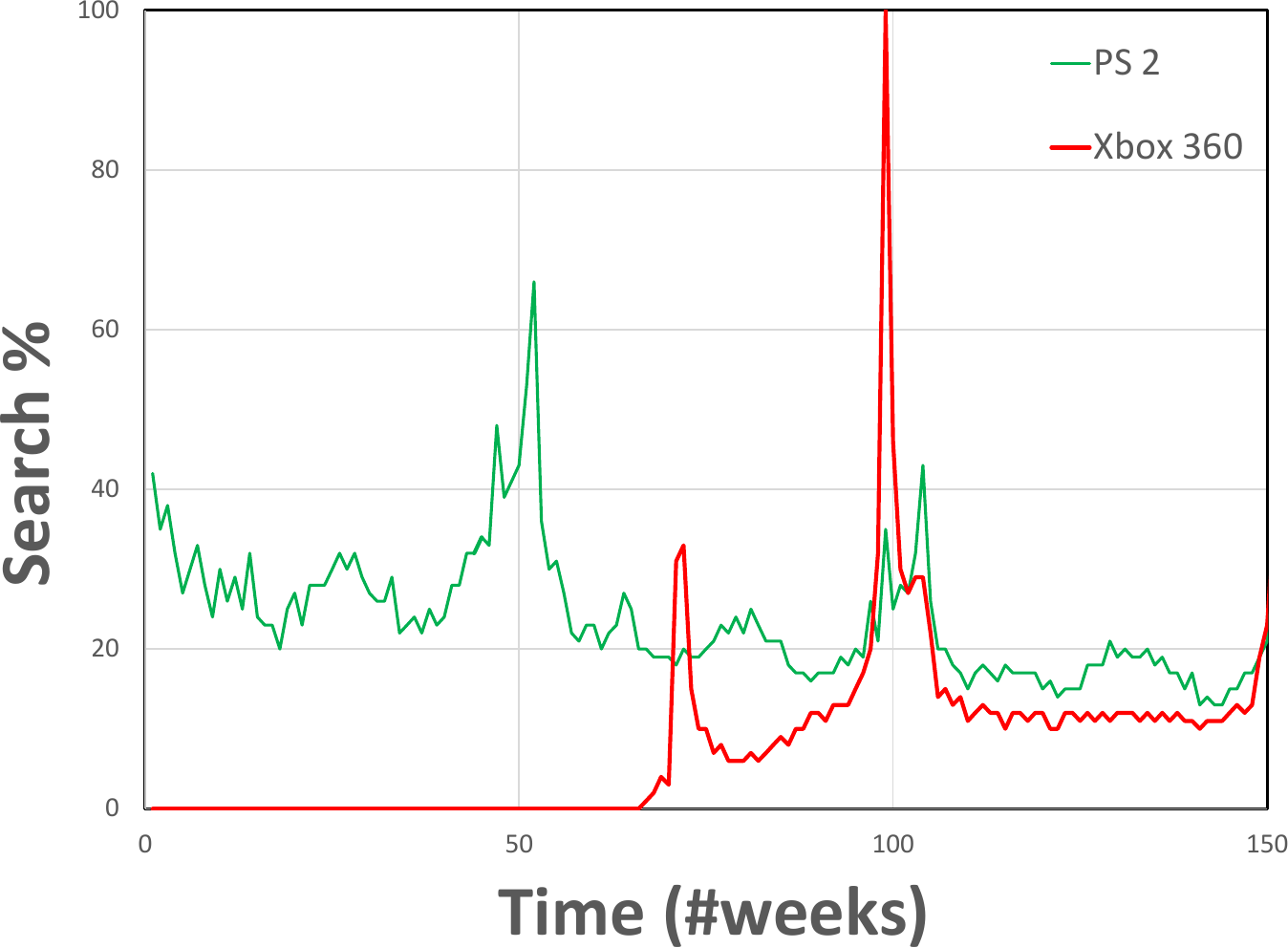} 
    \hfil
    \includegraphics[scale = 0.17]{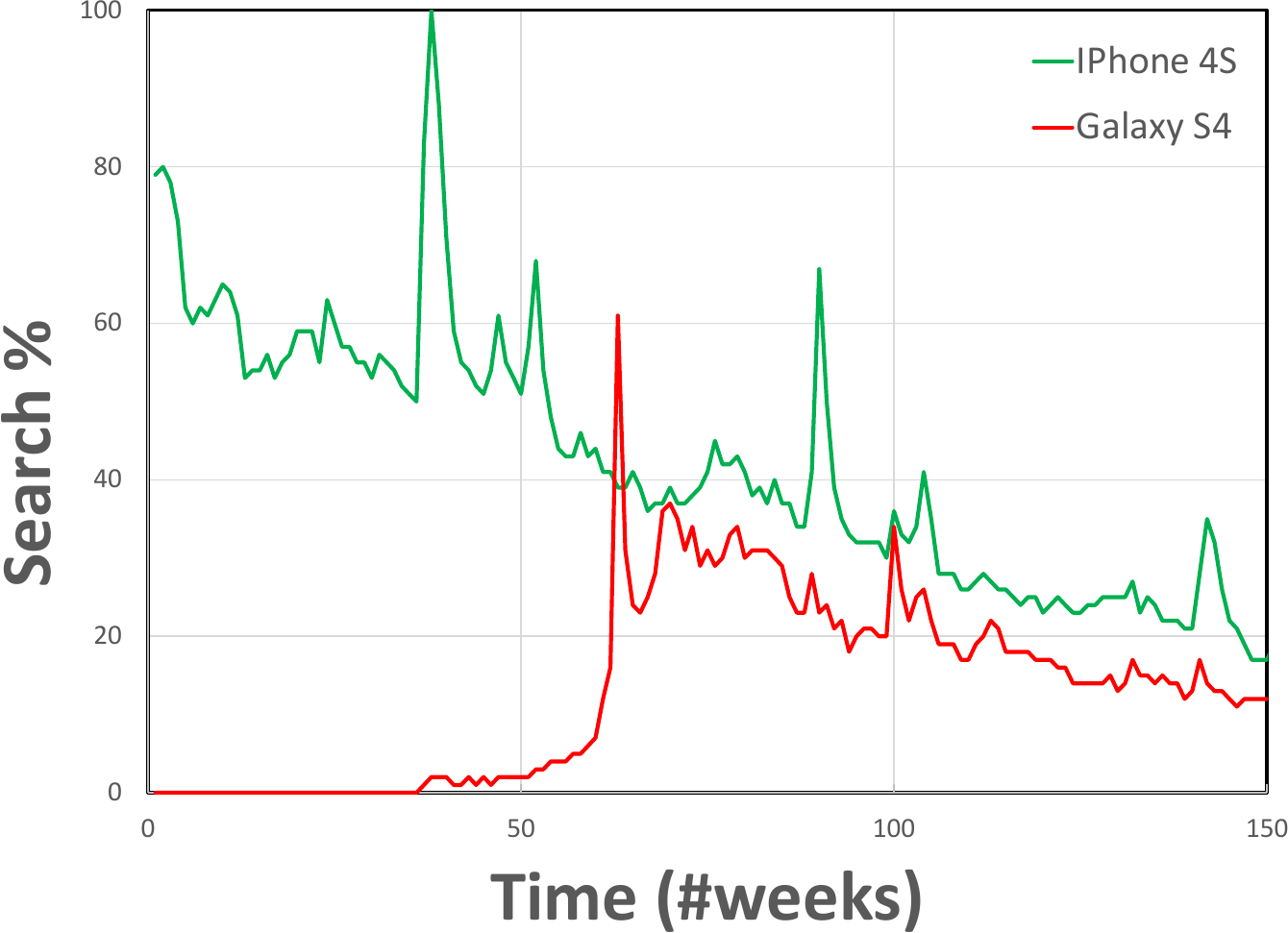}
    \caption{Real web-search percentage over a time period showing the introduction of a new product (Galaxy S4 \& Xbox 360) alongside an existing product (iPhone 4S \& PlayStation 2).}\label{fig:casestudy}
\end{figure} 
%----------------------------------------------------------------

Note that, in the above examples, the new product always enters a monopoly, i.e., market dominated by a single competitor, and manages to survive. Hence, given these motivating examples, we focus on the following scenario: any given market is dominated by a single product and a competing product strategizes (i.e., advertising to niche communities) to survive in that market. Specifically, our mathematical model captures the early phase of market scenario when there is only one product in the market and a `new' product tries to enter this monopoly.

When it comes to modeling product competition where individuals are the adopters, the interactions between them can be modeled as edges of a graph, which are linked via socially meaningful relationships like friendship or information exchange. Such social networks have become the paradigm for studying the spread of competing products and their adoption among individuals \cite{kalish1985new}. Although here we stress on `product' adoption and its influence among neighbors, we can draw a similar analogy within different contexts, for example, having an opinion (or preference) towards a political party/idea/information spreading in a network and how targeted advertisement plays an influential role in opinion formation/spreading mechanism. Given these motivating scenarios, we focus on settings where a new product enters a market dominated by a competitor. 

\textbf{Choice of Modeling Approach - bi-SIS:}  When two products compete with each other on a network, the eventual outcomes can be one of the following: absolute dominance, where one of the products dominates by wiping out its competitor and is adopted widely; die-out, where none of the products are adopted; or coexistence, where both of the products are adopted within the network \cite{prakash2012winner, beutel2012interacting, sahneh2014competitive, wei2016cooperative}. Various mathematical models, such opinion dynamics-based models (i.e., analysis of how opinions form/spread in a population) such as De-Groot, Hegselmann-Krause, Friedkin-Johnsen, variants of epidemic models \cite{omic2009epidemic, liu2019analysis}, and others \cite{proskurnikov2017tutorial, proskurnikov2018tutorial}, have been proposed to capture the dynamics of competing products. Most of the works around classical opinion dynamics-based models focus on showing eventual consensus (everyone agrees to one of the opinions/rumors eventually), which has also been extended to capture the scenario of coexistence of two opinions (or clusters of opposing opinions). These works have shown coexistence by either relying on external factors in the form of stubborn agents (initially or currently forming bias towards an opinion), zealots (who never alter their initial stance) or via the inclusion of confidence bounds beyond which individuals do not interact~\cite{gargiulo2008saturation,yang2012convergence, dandekar2013biased, mobilia2015nonlinear, amelkin2017polar, mukhopadhyay2020voter, anagnostopoulos2022biased, nguyen2020dynamics}. The addition of stubborn agents already guarantees that at least one of the opinions (the one followed by the stubborn agents) is always adopted and hence results in clusters of opinions. With the inclusion of these additional factors in the model setups, eventual consensus, or complete wipe-out is infeasible. 

The propagation of information, influence maximization, idea, or opinion/rumor diffusion in a network has also been extensively studied via Linear Threshold (LT)/ Information Cascade (IC) model~\cite{kempe2003maximizing, chen2010scalable}. In these models, individuals can be either in an active (i.e., influenced) or an inactive (i.e., free of influence) state. Here, the overarching assumption is that individuals either never change states or can only switch between active or inactive states randomly~\cite{granovetter1978threshold, pathak2010generalized}, with all nodes (users) becoming activated eventually. This is not always true in the context of products/opinions, as people tend to change their preferences/opinions/views from time to time. 

The Susceptible-Infected-Susceptible (SIS) epidemic model has been used to study the spread of viruses in a network of populations \cite{lajmanovich1976deterministic, van2008virus}. However, the \textit{bi-SIS} model, a variant of the SIS model (single virus), has been used to characterize the spreading dynamics of two competing products/opinions in a network \cite{prakash2012winner, sahneh2014competitive}. In such a model, via word-of-mouth or by being influenced by their neighbors, users adopt one of the two competing products (get infected) and gradually recover to become susceptible to infection again. Unlike the opinion dynamics-based model or LT/IC model, the bi-SIS model does not necessitate that the individuals have an opinion, positive or negative, ex-ante; in contrast, it allows individuals to switch states from being susceptible (inactive) to an infected (active) state and vice-versa. It has proved to be effective in capturing eventual consensus or coexistence in any general arbitrary graph, including special graphs such as complete or random graphs~\cite{wang2003epidemic, li2012susceptible, prakash2012winner, santos2015bi, yang2017bi}. It can be seen that the bi-SIS model aptly captures competitiveness without the need to have any prior assumptions or dependence on a specific network structure. Thus, the bi-SIS model captures all the three possible state spaces along with providing conditions for each of the outcomes i.e., die-out, absolute dominance, and coexistence. 

Although the bi-SIS model is effective in capturing all the possible outcomes, most of the existing works only focus on the stability analysis of the equilibrium points or identifying which of the three possible outcomes the system converges to \cite{van2008virus, li2012susceptible, prakash2012winner, sahneh2014competitive, doshi2021competing}. Furthermore, none of these works explicitly quantify the magnitude of existence; for example, the market share of either one or both of the competing products, especially in a general graph setting. Note that with respect to opinion dynamics, for example, the simple linear French-De-Groot model with stubborn agents \cite{french1956formal, proskurnikov2017tutorial} has shown quantifying results on a general graph. However, as mentioned earlier, unlike the bi-SIS model, the opinion dynamics-based models do not capture all the three possible state space and thus pose a limitation when it comes to the analysis of the market share of competing products. 

\textbf{Our Contributions:} In this paper, we take a step forward and propose techniques to quantify and improve the market share of the new/weaker product when it competes with an existing dominant product in the market via bi-SIS epidemic model. Under limited budget constraints, this new company launching its product would aim for survival, and the best strategy in this situation is to target certain users via advertisement/promotional offers to form a small support users' group (community) who, with their positive reviews, can spread about this new product. The motivation behind promoting such communities (additional edges among the users of the graph) is to increase the rate of influence so as to maximize the chance of its survival. Specifically, a resource-constrained company would want to identify an optimal set of individuals to form such communities to maximize its (positive) market share as much as possible. 

When it comes to quantifying the resulting market share on a general graph, there is no known closed-form for the positive fixed point (market share) even for a single-SIS model. The situation is further compounded when the competition comes into play (bi-SIS). To tackle this challenge, we employ the perturbation approach to optimally assign the resources under a specific budget structure based on model parameters and approximate the resulting market share of this original losing/weaker product. Our proposed budget structure also accounts for a minimum budget below which the company, despite spending positive amount of budget on advertisement expecting a non-zero return on investment, still gets zero market share and thus cannot compete with its competitor. By leveraging our theoretical findings, we develop a heuristic approach which returns a locally optimal choice of users, under a total budget constraint with heterogeneous costs for different users. We also demonstrate the efficacy of our approach by comparing against different benchmarks such as different standard centrality measures, over real network topology with various cost constraints for the users in the community. 

%----------------------SECTION 2------------------------------

\section{System Model and Problem Setup}\label{section2}

In this section, we take an initial look at the problem of optimally targeting individual users and recruiting them to participate in community, with the goal of maximizing the long term market share of the new product. It has been shown that when two products compete on a `same' graph\footnote{Same graph is analogous to one social network platform, for example, Facebook.} with same set of users, at steady state, the stronger one dominates and wipes-out the weaker one \cite{prakash2012winner}. WLOG, in our setting, on a `same' graph, \textit{Product 1} is the \textit{dominant} product in the market and \textit{Product 2} is the \textit{new/weaker} product (with zero market share) trying to compete against Product 1 for a positive market share. In our context, a community or online discussion forum provides a platform where a group of interacting users actively exchange information about common interests. Few pertinent examples can be a subreddit -- a community in Reddit~\cite{reddit} or influencer pods in social media platforms. In a subreddit, members can post information which is accessible to all members of the group, irrespective of whether members follow each other or not. In the case of influencer pods, groups of influencers work together to boost each other’s engagement (e.g., product's marketing)~\cite{grin}. We focus on the problem of \textit{creating} one such community. We assume that any interactions or discussions between users in this newly formed community are solely about the new company's product, and the influence among the users of the community is uniform/homogeneous.

%---------------------Images of Community------------------
\begin{figure}[!t]
    \centering
    \includegraphics[scale = 0.5]{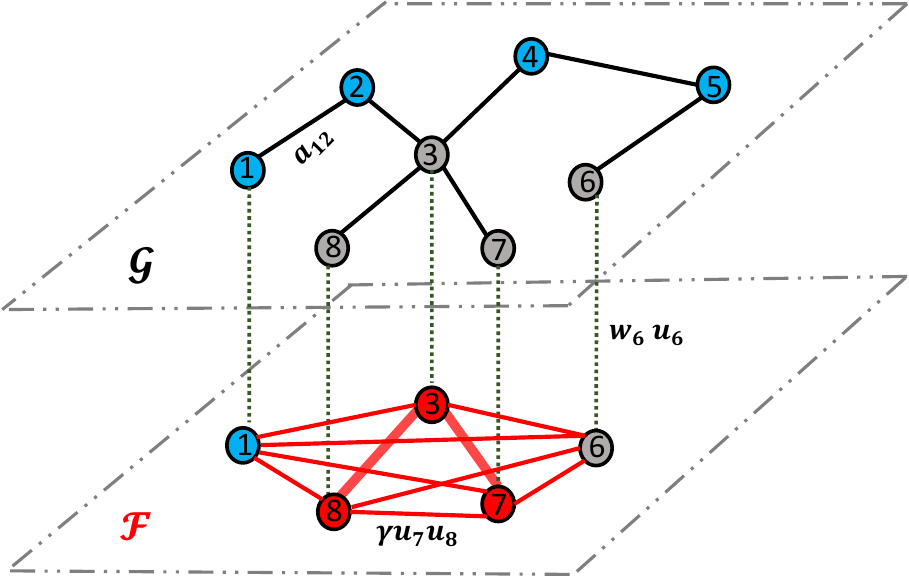}
    \caption{Illustration of users forming \textit{community} (Nodes 1,3,6-8 in layer $\cF$). The blue (red) nodes denote users adopting Product 1 (Product 2) at any given instance, and gray nodes are the users who do not own any of the products.}\label{fig:community figure}
\end{figure} 
%----------------------------------------------------------------
Suppose that any user whom the company has targeted in the form of advertisement/content-marketing chooses to engage in the community-based discussions with a certain probability. We can safely assume that the more aggressive the company is with its advertising efforts targeted towards the user, the higher is the probability of the user participating in this created community. In real life, these targeted advertisement efforts take various forms, such as product placement, early-bird discounts, or promotional offers, and cost a certain amount of resources. Later in this section, we give a more detailed discussion on the modelling aspects of the cost structure for our problem. 

\subsection{The System Model}

Let $u_i$ denote the probability with which any user $i \in \cN$ actively engages in such a community. We assume that all users participating in the community are able to communicate with one another, and Product 2's influence spreads across this community with rate $\gamma\beta_2$ (where $\gamma$ is a factor capturing the rate of influence relative to $\beta_2$). To clearly visualize this creation of a community, we draw it as a separate layer, as shown in Figure \ref{fig:community figure}. The top layer $\cG$ denotes the original graph where the mutual follower relationships are according to adjacency matrix $\mA$, and layer $\cF$ denotes the created community, in which a user $i\in\cN$ participates with probability $u_i$. The spread of influence within this community, is akin to adding an edge of weight $\gamma$ for any two users $i$ and $j$ \textit{known} to be actively participating in the community. Specifically, the parameter $\gamma$ denotes the amount the company is willing to spend on the members of the community to improve their communication efficiency (user’s interaction). If the two users were already connected prior to being a part of the community, that is, $a_{ij} = 1$ ($a_{12}$ in Figure \ref{fig:community figure}), then $\gamma$ is simply an additional reinforcement of their connection as a consequence of them interacting through multiple media. The thick `red' lines in Figure \ref{fig:community figure} highlight this reinforcement whereas the `black' lines indicate the original interaction links. Specifically, the introduction of additional edge weight (as shown in layer $\mathcal{F}$ in Figure \ref{fig:community figure}) is to induce an artificial community-like structure (clique) where the edge weight denotes how much the users are influenced from other users through their posts/shared information. The probability with which both $i$ and $j$ engage in the forum thus becomes $u_iu_j$ ($u_7 u_8$ in Figure \ref{fig:community figure}), and the total weight of edges connecting them is, on average, given by $a_{ij} + \gamma u_i u_j$. Table \ref{table:symbols} explains the terminology we have used in this paper. We denote vectors ($\vecv \!\in\! \R^{N}$) and matrices ($\mH \in \R^{N \times N}$) as bold lowercase and uppercase letters, respectively. Note that in our model we assume that the Product 1 does not advertise to counter this `incursion' by Product 2 and a user cannot adopt both products at the same time. We also assume that the complete network structure information is available to the company before launching the new product. 

%----------------------------------------------------------------
\renewcommand{\arraystretch}{1.1}
\begin{table}[]
\resizebox{\columnwidth}{!}{%
\begin{tabular}{|l|l|}
\hline
\multicolumn{1}{|c|}{\textbf{Symbol}}        & \multicolumn{1}{c|}{\textbf{Definition/Description}}              \\ \hline 
$\mathcal{G(N,E)}$                          & Undirected, connected graph                                      \\ \hline
$\cN = \{1,2,\dots, N\}$                     & Set of nodes                                                      \\ \hline 
$\mA \triangleq [a_{ij}]$                    & Adjacency matrix of graph $\cG$                                   \\ \hline 
$\beta_1$ ($\beta_2$)                        & Adoption rate of product 1 (or product 2)                         \\ \hline 
$\delta_1$ ($\delta_2$)                      & Disown rate of product 1 (or product 2)                           \\ \hline 
$\tau_1 = \beta_1/\delta_1$ ($\tau_2$)       & Effective adoption rate of product 1 (or product 2)               \\ \hline 
$u_i$                                        & Probability with which any user $I$ actively engages in community \\ \hline 
$w_i$                                        & Cost of recruiting user $i$ in community                          \\ \hline 
$\gamma$                                     & Influence rate among community users (uniform)                             \\ \hline 
$C$                                          & Budget allocated by company launching new product (product 2)     \\ \hline 
$\lambda = \lambda(\mH) > 0$                 & Largest eigenvalue of matrix $\mH$                                \\ \hline 
$\vx^* (\vy^*) \gg \0$                       & Positive fixed points corresponding to single-SIS                 \\ \hline 
$\mS_{\vecv} \triangleq diag(\ones - \vecv)$ & Diagonal matrix with elements $(1 - v_i)$                         \\ \hline 
\end{tabular}%
}
\caption{Symbols}
\label{table:symbols}
\end{table}
%----------------------------------------------------------------

As mentioned earlier in Section \ref{Introduction}, we use the bi-SIS ODE system to model our system dynamics, which yields the following
%================================================================
\begin{equation}
  \begin{aligned}\label{eq:ode_gammaUU}
     \dot x_i(t) &= \beta_1 (1\!-\!x_i(t)\!-\!y_i(t))\sum_{j \in \cN} a_{ij} x_j(t) - \delta_1 x_i(t) \\
     \dot y_i(t) &= \beta_2 (1\!-\!x_i(t)\!-\!y_i(t))\sum_{j \in \cN}[ a_{ij} +\gamma u_i u_j ]y_j(t) \\
     &- \delta_2 y_i(t),
  \end{aligned}
\end{equation}
%================================================================
where $x_i(t)$ ( $y_i(t)$ ) is the probability with which user $i$ owns Product 1 (Product 2) at time $t \geq 0$. Also, $\beta_2(a_{ij} + \gamma u_i u_j)$ is the overall expected rate at which the user $j$, say who already owns Product 2, causes the susceptible (owns neither products) user $i$ to adopt Product 2. Note that \eqref{eq:ode_gammaUU} is an extension of the original bi-SIS model \cite{prakash2012winner} whose system dynamics for any two competing products propagating on a graph with adjacency matrix $\mA$ is given by
%================================================================

\begin{align*}
    \dot x_i(t) &= \beta_1 \big(1- x_i(t) - y_{i}(t)\big) \sum_{j \in \cN} a_{ij} x_{j}(t) - \delta_1 x_{i}(t)\\
    \dot y_i(t) &= \beta_2 \big(1- x_{i}(t) - y_{i}(t)\big) \sum_{j \in \cN} a_{ij} y_{j}(t) - \delta_2 y_{i}(t).
\end{align*}

%================================================================
It has been established that the above system dynamics results in only the stronger product to dominate the market and wipe out the weaker one. We, in this paper, try to bridge this gap by inducing a community ($\gamma u_i u_j$ in \eqref{eq:ode_gammaUU}) within the original graph so that the weaker product can survive (i.e., a positive market share) in the market.

Using $\vu = [u_i]_{i \in \cN}$ to denote the vector of probabilities with which users participate in the community, we can write \eqref{eq:ode_gammaUU} in a matrix-vector form as follows:
%================================================================
\begin{equation}
  \begin{aligned}\label{eq:ode_gammaUU-vector}
     \dot \vx(t) &= \beta_1 \text{diag}(\ones \!-\!\vx(t)\!-\!\vy(t))\mA\vx(t) - \delta_1 \vx(t) \\
     \dot \vy(t) &= \beta_2 \text{diag}(\ones\!-\!\vx(t)\!-\!\vy(t))\left[ \mA+\gamma \vu \vu^T \right]\vy(t) \\
     & -\delta_2 \vy(t).
  \end{aligned}
\end{equation}
%================================================================
As can be seen in \eqref{eq:ode_gammaUU-vector}, the extra edge weights added to the original adjacency matrix $\mA$, as a result of creating community in the manner as described earlier, are in the form of a rank one update given by $\mA + \gamma\vu\vu^T$.\footnote{$\gamma\vu\vu^T$ leads to a clique with homogeneous influence rate among the users. Note that heterogeneous influence rate (while ensuring more generality) introduces further complications in the analysis, which we leave as a future work.} Later in Section \ref{section3}, we use this special form of the update to aid our analysis. Our propagation model \eqref{eq:ode_gammaUU-vector} yields the equilibrium equations
%-----------------------------------------------------------------
\begin{align}
    \vx &= \tau_1\text{diag}(\ones - \vx - \vy)\mA \vx   \label{eq:fpe x}\\
    \vy &= \tau_2\text{diag}(\ones - \vx - \vy)\left[\mA + \gamma \vu \vu^T \right]\vy,   \label{eq:fpe y}
\end{align}
%-----------------------------------------------------------------
where $\vx$ and $\vy$ are the fixed points of~\eqref{eq:ode_gammaUU-vector} for Product 1 and Product 2, respectively.

\subsection{Problem Setup}
We present the problem of maximizing the market share $\bar{y} = 1/N \sum_{i \in \cN} y_i$ (where $\vy = [y_i]$ is the fixed point for Product 2 from \eqref{eq:fpe y}) under a limited budget constraint, say $C$, by the optimization problem given by
%----------------------------------------------------------------

\onetagleft
\begin{align}\tag{$\mathcal{M}$}\label{op:maximize y}\hspace{2em}
\underset{\gamma, \vu}{\text{maximize}}\ \ \ \  &\bar y = \frac{1}{N}\!\sum \limits_{i \in \cN} y_i \nonumber\\ 
\text{subject to}\ \ \  \ &\sqrt{\gamma}\sum \limits_{i \in \cN} w_i u_i \leq C \nonumber\\
                          &u_i \in [0,1] \ \ \ \ \forall i\in\cN. \nonumber
\end{align}

%----------------------------------------------------------------

We structure the cost constraint in this optimization problem in a heterogeneous manner. Under this setting, $w_i$ is the cost of recruiting user $i$ in a newly created community. This user $i$ will be included in the forum with probability $u_i$ such that the `average cost' of having user $i$ in the forum would be $w_i u_i$ ($w_6 u_6$ in Figure \ref{fig:community figure}). This setting is very general and can capture various scenarios such as $w_i$ being larger (more expensive to recruit) for more popular users/influencers as would be the case in reality. The company also allocates a portion of its budget towards $\gamma$, which captures the overall communication efficiency across the channels established for user interaction, as part of creating the discussion group\footnote{The parameters $\gamma$ and $w_i$ set in this model are part of the budget, which in a practical situation would be determined as a function of the total budget allocated by the company launching the new product.}. We reflect this in the budget constraint in our optimization problem \ref{op:maximize y}, with the presence of the $\sqrt{\gamma}$ term. This ensures that the scaling between $\gamma$ and $\vu$ is the same as $\gamma \vu \vu^T$ as in \eqref{eq:fpe y}. To be more specific, for any choice of $\vu$ up to a multiplicative constant, we can allocate $C$ between $\gamma$ and $\vu$ without affecting the fixed point $\vy$. Our scaling ensures that $\hat \vu = c\vu$ and $\hat \gamma = \gamma/c^2$ for any constant $c$ renders the terms $\gamma \vu \vu^T$ unchanged. Any other scaling would force us to consider every possible value of $\vu$ and $\gamma$ \emph{separately} and thus obfuscate the meaning of the budget constraint for our optimization problem and subsequently our analysis on the optimal choice of $(\gamma, \vu)$ on the maximal market share $\bar{y}$. This completes our problem set up, and we are now ready to address its analysis in the next section. 

%------------------------SECTION 3-------------------------------

\section{Analysis of Product 2's Market share $\bar{y}$}\label{section3}
In this section, we begin by first analyzing the budget constraint in \ref{op:maximize y} and the minimum amount of budget needed to ensure positive market share of Product 2, $\bar{y}>0$. We also establish a key result on the behavior of $\bar{y}$ as the entries $u_i$ of the vector of probabilities $\vu$ (probability with which any user $i \in \cN$ actively engages in the community) are perturbed. We then define an optimization problem that helps in updating $\vu$, consequently helping towards obtaining an improved market share of Product 2 under the same budget constraint. We defer to the Appendix all the technical proofs of the results presented in this section.

\subsection{Minimum Budget for $\bar{y}>0$}\label{sec:budget discussion}

To begin with, we provide the following proposition regarding the budget constraint in \ref{op:maximize y}.

%========================================================
\begin{proposition}\label{prop::budget equality}
Any solution $(\gamma^*,\vu^*)$ to the optimization problem \ref{op:maximize y} always satisfies the budget constraint with equality, that is,
$$\sqrt{\gamma^*}\sum \limits_{i \in \cN} w_i u_i^* = C.$$
\end{proposition}
%=============================================================

Proposition \ref{prop::budget equality} tells us that the optimal market share of Product 2 is always obtained by spending all the available budget. This follows from the monotonicity property\footnote{Arising from the fact that the bi-SIS system is a \textit{monotone dynamical system} \cite{doshi2021competing}.} of the ODE system \eqref{eq:ode_gammaUU-vector} in that $\bar{y}$ turns out to be an increasing function of $\gamma$ and $\vu$, albeit being nonlinear. See Appendix \ref{appendix:prop budget equality} for details. 

For general bi-SIS models, Product 2 needs to satisfy a threshold type survival condition $\tau_2 \lambda(\mS_{\vx^*}[\mA + \gamma \vu \vu^T]) > 1$, in order to secure a positive market share in the long run\footnote{Both products survive (coexistence equilibria) when $\tau_1 \lambda(\mS_{\vy^*}\mA)\!>\!1$ and $\tau_2 \lambda(\mS_{\vx^*}[\mA + \gamma \vu \vu^T])\!>\!1$, and the system globally converges to $(\vx, \vy)\!\gg\!(\boldsymbol \0, \boldsymbol \0)$. Here, $\vx, \vy$ correspond to coexistence fixed points \cite{sahneh2014competitive, doshi2021competing}.}. This condition captures the relationship between the initial market share of Product 1, given by $\bar{x}^* = 1/N \sum_{i \in \cN} x^*_i$, and the budget $C$, which is a function of the pair $(\gamma, \vu)$. Specifically, for a given $\tau_2$ and any $\vu$ fixed up to a multiplicative constant, the survival condition indicates that when $\bar{x}^*$ is large, the term $\gamma \vu \vu^T$ also needs to be large (to ensure a large enough eigenvalue)\footnote{This follows from the fact that the spectral radius of a non-negative matrix is a non-decreasing function of its elements \cite{meyer2000matrix} and the terms in $\mS_{\vx^*} = \text{diag}(\ones - \vx^*)$ is decreasing in $\vx^*$.} in order for Product 2 to survive in the long term, subsequently needing a large budget to do so. This implies that as the market share of Product 1 increases, the minimum budget required for Product 2 to survive also increases, which can also be seen from the heat-map as depicted\footnote{We numerically simulate the ODE system \eqref{eq:ode_gammaUU-vector} for a small graph of dolphins' population (62 nodes, 159 edges) \cite{lusseau2003bottlenose}.} in Figure \ref{fig:heatmap}. For an arbitrarily chosen $\vu$, fixed up to a multiplicative constant, Figure \ref{fig:heatmap} illustrates how the steady-state market share of Product 2 varies with $C$ and $\vx^*$. The `red' curve in Figure \ref{fig:heatmap} indicates the minimum budget required to attain $\bar{y} > 0$, and any amount spent below the minimum budget leads to a zero market share. In other words, a positive budget spent by the company towards its product placement does not necessarily guarantee even a small non-zero return; the company launching a new product must account for the existing competitor's market share to decide the minimum budget to ensure its non-zero long-term market share. 

%------------------------------------------
\begin{figure}[!t]
    \centering
    \includegraphics[scale=0.55]{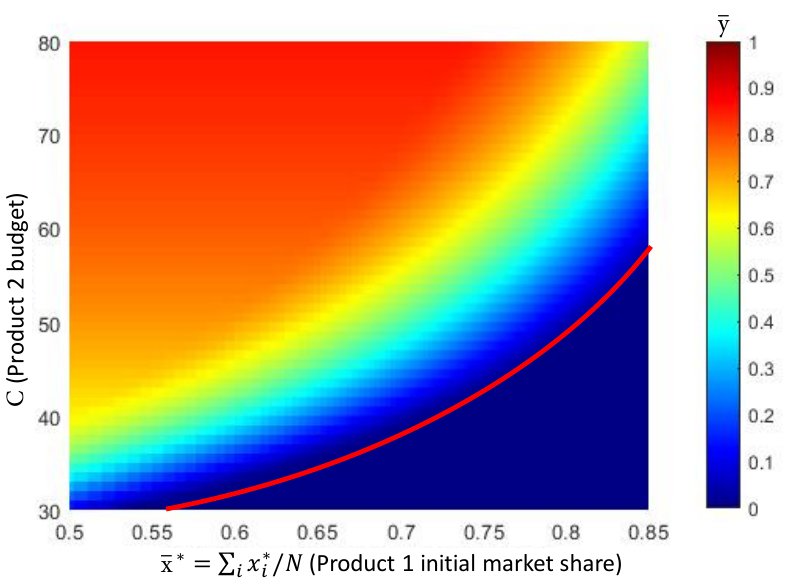}
    \caption{Heat map illustration of the market share as a function of budget $C$ and $\bar{x}^*$.}
    \label{fig:heatmap}
\end{figure}
%------------------------------------------
\subsection{Main Result}\label{subsec:sensitivity ybar}

We now turn our attention towards establishing a key result on the sensitivity of $\bar{y}$, with respect to entries $u_i$ of $\vu$. For a general graph with adjacency matrix $\mA$, it is impossible to obtain the optimal solution to \ref{op:maximize y} in a closed form; the main issue being the absence of a closed-form expression of $\bar{y}$. Even for the single-SIS model in a general graph setting, there is no known closed-form for the positive fixed point. The situation is further compounded when competition comes into play in the bi-SIS model. Therefore, to gain better insight into $\bar{y}$ and how it would change with model parameters, we begin by applying small perturbation to $u_i$ for each $i \in \cN$. As previously mentioned, change in $\gamma$ does not affect the fixed point in \eqref{eq:fpe x}, \eqref{eq:fpe y} when the budget is kept constant, since the term $\gamma \vu \vu^T$ remains the same under the aforementioned scaling for $\gamma$ and $\vu$ with the same budget $C$. The only way to obtain an increase in $\bar{y}$ is by \emph{heterogeneously} changing the entries $u_i$, and thereby obtaining an $\emph{updated}$ version which is no longer a multiplicative constant of the initial $\vu$.

To proceed, we first define the terms \textit{critical curve} and \textit{critical parameters}, which we use throughout the rest of this section. The `critical curve' is the set of parameters at which the aforementioned survival condition is satisfied with equality, and we call any such set of parameters `critical' if they are on the critical curve. Let $(\gamma^c,\vu^c)$ denote a pair of such critical parameters for which $\tau_2 \lambda(\mS_{\vx^*}[\mA + \gamma^c \vu^c {\vu^c}^T]) = 1$. Lemma \ref{lemma::initial_candidate} provides a feasible `initial' pair of critical parameters.

%========================= Lemma 1 ===========================
\begin{lemma} \label{lemma::initial_candidate}
Let $\vecv \in \R^N$ be the PF eigenvector of \hspace{0.2em}$\mS_{\vx^*}\mA$ associated with $\lambda(\mS_{\vx^*}\mA)$. Then, for the choice of $\vu^c$ and $\gamma^c$ such that
$$\vu^c \propto \mS_{\vx^*}^{-1} \vecv, ~~\text{  and  }~~ \gamma^c = \frac{1/\tau_2  - \lambda(\mS_{\vx^*}\mA)}{{{\vu^c}^T}\mS_{\vx^*}\vu^c},$$we have $\tau_2 \lambda (\mS_{\vx^*} [\mA+ \gamma^c \vu^c \vu^{cT}])=1$.
\end{lemma}
%===============================================================

For any given such pair, let $\vecnu^c = [\nu^c_i]$ be the PF\footnote{The Perron-Frobenius (PF) eigenvector is the eigenvector corresponding to the largest eigenvalue of a non-negative, irreducible matrix \cite{meyer2000matrix}.} eigenvector of $\mS_{\vx^*}[\mA + \gamma^c \vu^c {\vu^c}^T]$. From the threshold condition we know that $\bar{y} = 0$ when $\tau_2 \lambda(\mS_{\vx^*}[\mA + \gamma^c \vu^c {\vu^c}^T]) = 1$. As a first step towards obtaining a positive $\bar{y}$, we apply perturbation $\alpha_i >0$ ($\vecalpha = [\alpha_i]$) to any $i$'th entry $u_i^c$ of $\vu^c$. The updated entry is then given by $u^c_i + \alpha_i$, and we have the following result.

%===================== Theorem =================================
\begin{theorem}\label{theorem::avg_y}
The market share of Product 2 as a function of the perturbation $\alpha_i >0$ to $u_i^c$, can be written as
\begin{equation}\label{eq: them avg y}
 \bar{y}(\alpha_i) = \zeta \nu_i^c \alpha_i + O(\alpha_i^2 ),
\end{equation}
where $\zeta >0$ is a constant independent of $i$.
\end{theorem}
%================================================================

\noindent Theorem \ref{theorem::avg_y} tells us that the magnitude of increase in $\bar{y}$ due to increase in $u^c_i$ (that is, $\alpha_i$) is roughly proportional to $\nu^c_i$ with an error term\footnote{Note that Theorem \ref{theorem::avg_y} shows the sensitivity of $\bar{y}$ with respect to `one' particular $\alpha_i >0$ for which we expect to see the best possible increase in the market share. Any $\alpha_i < 0$ results in a zero market share $\bar{y} = 0$ since $\tau_2 \lambda(\mS_{\vx^*}[\mA + \gamma^c \vu^c {\vu^c}^T]) \leq 1$.}. This implies that there is a greater reward in terms of market share to invest in nodes with larger $\nu_i^c$. The result from Theorem \ref{theorem::avg_y} helps in obtaining a direction leading to an increase in $\bar{y}$. 

Suppose we have an initial pair of model parameters which satisfy the threshold condition, and thus form a critical pair $(\gamma^c, \vu^c)$. These chosen parameters would then satisfy the budget $C = \sqrt{\gamma^c}\sum_{i \in \cN} w_i u_i^c$. By keeping the budget constant, we intend to move in a direction that obtains the best possible increase in $\bar{y}$ from its initial value of $\bar{y}=0$ at the critical point. We do so by applying a perturbation to $u_i^c$, which needs to have both positive and negative components to ensure that the budget constraint remains binding. Also, Theorem \ref{theorem::avg_y} indicates that the increase in $\bar{y}$ is roughly proportional to $\nu_i^c$. 

We leverage this result and the perturbation approach to formulate another optimization problem \ref{op:maximize alpha}. Since Theorem \ref{theorem::avg_y} holds in a small region around the critical curve, we limit the maximum size of perturbation to $\epsilon > 0$. Specifically, for a given initial $\vu^c$, we define $\tilde u_i \triangleq u_i^c + \alpha_i$, for all $i \in \cN$ where $\alpha_i \in [-\epsilon, \epsilon]$ and the step size $\epsilon$ is a small constant that controls the maximal size of perturbation in $\vu$. We then consider the following optimization problem:

%-----------------------------------------------------------------
\onetagleft
\begin{align} \tag{$\mathcal{P}(\vecnu,\epsilon,\vu)$} \label{op:maximize alpha}
\hspace{5em} 
 \underset{\alpha_1, \cdots, \alpha_N}{\text{maximize}} \ \ \ \ &\underset{i \in \cN}{\sum}\alpha_i \nu_i  \nonumber\\
\text{subject to} \ \ &\sum_{i \in \cN} w_i \alpha_i = 0, \nonumber \\
\ \ &\alpha_i \in [-\epsilon, \epsilon]. \nonumber 
\end{align}

%---------------------------------------------------------------

The solution $\vecalpha$ to \ref{op:maximize alpha} is based on the decreasing order of $(\nu_i/w_i)$ such that it satisfies the given constraints. This solution is essentially a relaxation of the well-known \textit{0-1 Knap-sack} problem \cite{garey1979computers}, which for a given initial critical parameters $(\gamma^c, \vu^c)$ leads to a direction resulting in the best possible increase in $\bar{y}$. The perturbed $\vu$ is the optimal number of users who can be targeted with advertisements of Product 2, and thus results in a positive market share.

%----------------------SECTION 4--------------------------------

\section{Experimental Setup \& Results}\label{section4}

In this section, we present simulation results on real-world graphs under a cost constraint for a given initial ($\gamma, \vu$) pair according to Lemma \ref{lemma::initial_candidate}. We first briefly describe the simulation setup for our experiments, which includes the budget structure and then describe the baseline methods that we use to compare its performance against our approach.

\subsection{Simulation Setup}
We use the following two real-world social network based graph datasets for the simulations, which are undirected, connected graphs from SNAP \cite{snapnets} and Network \cite{nr} repositories.
\begin{enumerate}
    \item Social circles- Facebook: Friends list of users using Facebook App, which contains 4,039 nodes and 88,234 edges \cite{mcauley2012learning}.
    \item Facebook pages of public figures: Nodes which are blue-verified public figures' facebook pages and edges are mutual likes among them. It contains 11,565 nodes and 67,000 edges \cite{feather}.
\end{enumerate}
As Product 1 already exists in the network and being stronger than Product 2, the effective adoption rates are such that initially $\tau_1 \lambda(\mS_{\vy^*}\mA) \!>\!1$ and $\tau_2 \lambda(\mS_{\vx^*}\mA)\!\leq\!1$ (zero market share). Here, $\tau_1 = 0.8, \tau_2 = 0.05$ for Facebook dataset with 4039 nodes and $\tau_1 = 0.4, \tau_2 = 0.03$ for Facebook dataset with $\approx 12k$ nodes. In the simulation results, we are interested in the immediate increment of $\lambda(\mS_{\vx^*}\mB)$ (where $\mB = \mA + \gamma \vu \vu^T$) from its initial value $\tau_2 \lambda(\mS_{\vx^*}\mB) = 1$ for a given initial pair $(\gamma^c, \vu^c)$. 

We set the budget as per Proposition \ref{prop::budget equality} and run our simulation for different values of $\epsilon$ (parameter that controls the perturbation of $\vu$) under two scenarios: (i) when the cost of resource allocation is homogeneous across all nodes, for example, $w_i = 1$ for all $i \in \cN$ and, (ii) when the cost of resource allocation is heterogeneous, for example, one instance of $\vw$ sampled uniformly at random. The proposed algorithm for obtaining the locally optimal set of users is given in Algorithm \ref{algo:our_algo}.

\subsection{Comparison with Baseline Methodologies}

%---------------------------Algorithm1---------------------------
\setlength{\textfloatsep}{4pt}
\begin{algorithm}[!t]
    \DontPrintSemicolon
    \SetAlgoLined
    \SetKwComment{Comment}{/* }{ */}
    \SetKwProg{Init}{Initialize}{:}{}
    \caption{Local Search for bi-SIS}
    \label{algo:our_algo}
            \KwIn{$\mA, \tau_1, \tau_2, \vw, \epsilon$}
            \KwOut{$\vu$}
            
            \Init{}{
                $\vecv(0) = \vecnu^c$\Comment*[r]{$\mS_{\vx^*}\mA$'s eigenvector}
                $\vu(0) = \vu_0$\Comment*[r]{as in Lemma \ref{lemma::initial_candidate}}
                $\gamma(0) = \gamma_0$\Comment*[r]{as in Lemma \ref{lemma::initial_candidate}}
                $\epsilon(0)=\epsilon$\;
                %$\tilde \gamma(0) = \gamma(0)$\;
                }

                $\epsilon(1) \leftarrow \min\{\epsilon(0), \min\{\vu(0)\}\}$\;
                $\boldsymbol \alpha \leftarrow \mathcal{P}(\vecv(0),$ $\epsilon(1), \vu(0))$\Comment*[r]{Solution to optimization problem}
                $\vu(1) \leftarrow \vu(0) + \boldsymbol \alpha$\; 
                %\SetKwProg{set}{Set}{:}{}
                %\set{} {
                %}\;

            %\KwRet{$\vu$}
\end{algorithm}
%---------------------------------------------------------------

To evaluate the performance, we compare our proposed approach of selecting a community against two methodologies: (i) standard centrality measures such as degree-based centrality and eigenvector-based centrality, and (ii) NetShield algorithm \cite{chen2015node}, as baseline algorithms. In this section, we describe each of the methodologies along with the reasoning for selecting these as baseline algorithms.

Degree-based centrality is a measure of node connectivity based on the number of links (degree) each node has with others which is useful in identifying popular nodes in a network who can further connect to a wider network \cite{newman2018networks}. Similar to this, eigenvector-based centrality (EVC) measures a node's influence that it has on other nodes in the network. This is done by assigning an importance score based on the fact that a node is important when it is linked to another important node. EVC identifies not only the influence of a node on a direct connection, but also the impact on the entire network \cite{wasserman1994social}. We chose to compare against these measures to gain a holistic analysis against such well-known centrality methods which are generally used to understand the relationship structure in a social network \cite{freeman1978centrality, friedkin1991theoretical}. Also, one of the reasons to choose these centrality measures is because one would assume that spending resources on influential users (users with higher degree or evc score) would return higher market share, which we show otherwise in our simulation results. 

From the methodological perspective, Chen et. al.~\cite{chen2015node} uses a very similar epidemic modeling approach and have proposed an algorithm `NetShield' that helps in determining the number of users/nodes who can be immunized to contain the epidemic. Within our context, we utilize this algorithm as one of the baselines to determine the number of users who can be selected to form the community. In \cite{chen2015node}, NetShield algorithm aims to minimize the influence spread by first selecting and then deleting those nodes who are more likely to further spread the influence in the graph. For baseline comparison, we utilize this algorithm to determine those highly influential users who can be recruited to form a community who can then increase the influence spread of the new product in the graph. 
%----------------------------------------------------------------
\begin{algorithm}[!t]
    \DontPrintSemicolon
    \caption{Centrality Measure Methodology}
    \SetKwComment{Comment}{/* }{ */}
    \label{algo:bm}
    %\begin{multicols}{2}
            \KwIn{$\mA, \tau_1, \tau_2, \gamma^c, N, C, \vw, \mu$}\Comment*[r]{$\mathbf{\mu}$ = chosen centrality measure}
            \KwOut{$\vu$}
                    $u_i = \frac{C \mu_i w_i}{\sum \limits_{j \in T}\mu_j w_j}$ for all $i \in N$\;
                    
    %\end{multicols}
 \end{algorithm}
%----------------------------------------------------------------
In each of these baseline methodologies, while the total budget remains the same, that is, $C = \sum_i \sqrt{\gamma} w_i u_i$, we assign the cost $w_i$ (either homogeneous or heterogeneous) based on the ratio of its degree and importance score for degree and eigenvector centrality measure, respectively as outlined in Algorithm \ref{algo:bm}. Similarly, the cost $w_i$ is assigned to the users obtained from the NetShield algorithm. The obtained $\vu$ is then similarly fed to \eqref{eq:fpe y} to obtain the market share of Product 2.

\subsection{Simulation Results}

Here, we present the comparison results for both the homogeneous and heterogeneous cost distribution for the above mentioned real-world graphs dataset. To visualize our system in the two dimensions, we represent Product 2's (Product 1) market share on the y-axis as \textit{AvgY}$= (1/N) \sum_{i \in \cN} y_i$ (\textit{AvgX}$= (1/N) \sum_{i \in \cN} x_i$). 

As mentioned in Section \ref{subsec:sensitivity ybar}, the step size $\epsilon$ is a small constant that controls the maximal size of perturbation in $\vu$. The obtained $\vu$ is then fed into the bi-SIS ODE \eqref{eq:fpe y} along with other associated parameters which is then simulated over 1000 time steps to obtain the market share of Product 2 (Similarly for Product 1 using \eqref{eq:fpe x}). Figures \ref{fig:diff_epsilon_fbhomo} -- \ref{fig:diff_epsilon_fbpublichetero} show how the market shares of Product 1 and Product 2 change based on different $\epsilon$ values when the cost is homogeneously and heterogeneously distributed. The results indicate that as the $\epsilon$ value increases, the market share of Product 2 also increases while the market share of Product 1 decreases, under both homogeneous and heterogeneous cost distribution. This demonstrates that the users' probability of willingness to participate in the community also increases as the $\epsilon$ value increases, thereby resulting in a higher market share. 

%------------------------------------------
\begin{figure}[!t]
    \centering
    \includegraphics[width = 0.78\columnwidth]{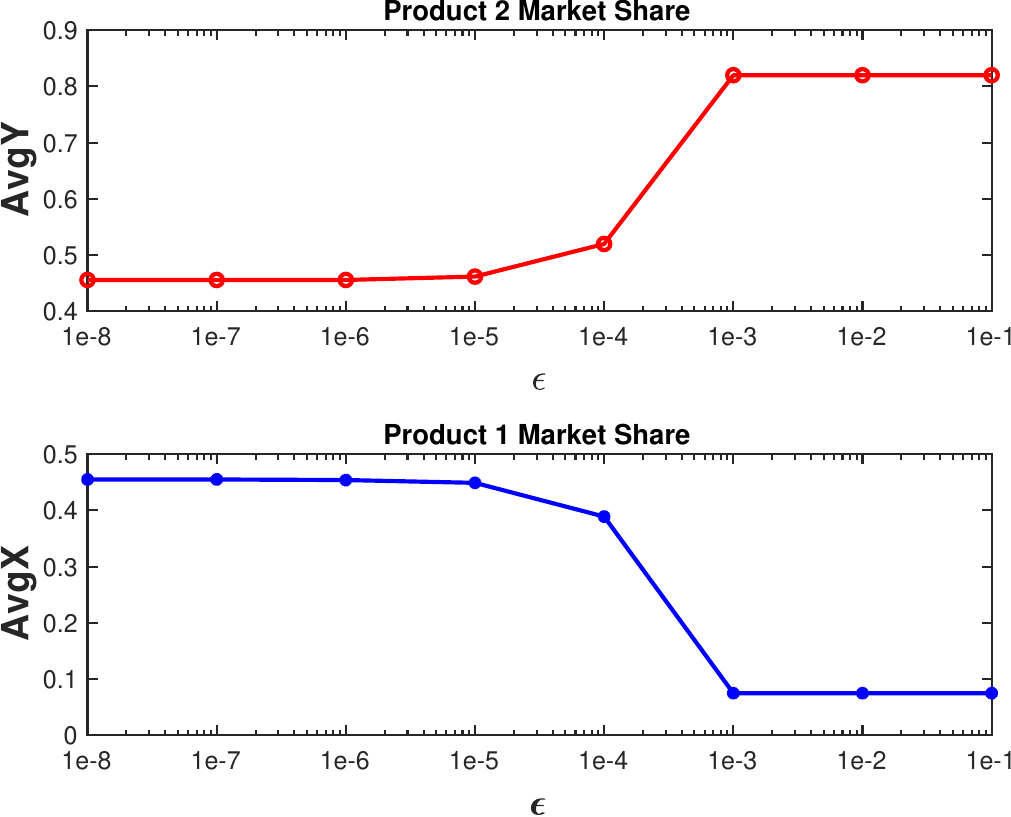}
     \caption{Facebook ($\approx4k$ nodes): Change in $\bar{\vy}$ (AvgY) and $\bar{\vx}$ (AvgX) for \textit{Homogeneously} distributed cost and different values of $\epsilon$.}
     \label{fig:diff_epsilon_fbhomo}
\end{figure}
%------------------------------------------

%------------------------------------------
\begin{figure}[!h]
    \centering
    \includegraphics[width = 0.78\columnwidth]{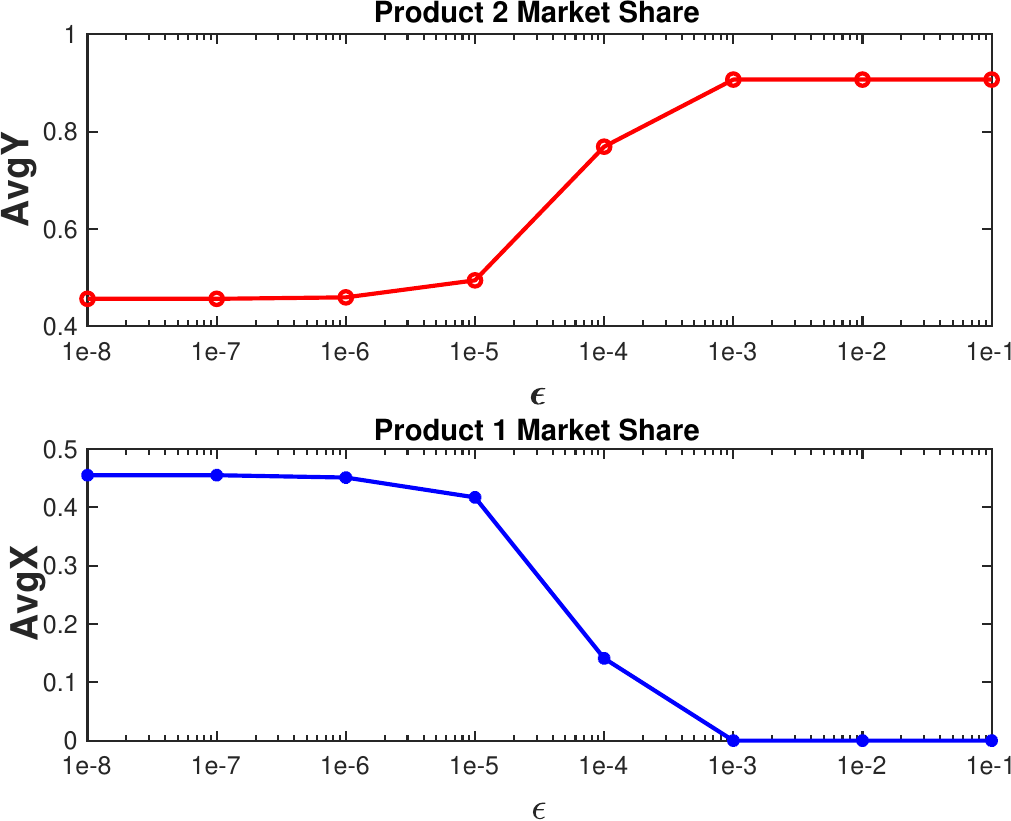}
    \caption{Facebook ($\approx4k$ nodes): Change in $\bar{\vy}$ (AvgY) and $\bar{\vx}$ (AvgX) for \textit{Heterogeneously} distributed cost and different values of $\epsilon$.}
    \label{fig:diff_epsilon_fbhetero}
\end{figure}
%------------------------------------------

%------------------------------------------
\begin{figure}[!h]
    \centering
    \includegraphics[width = 0.78\columnwidth]{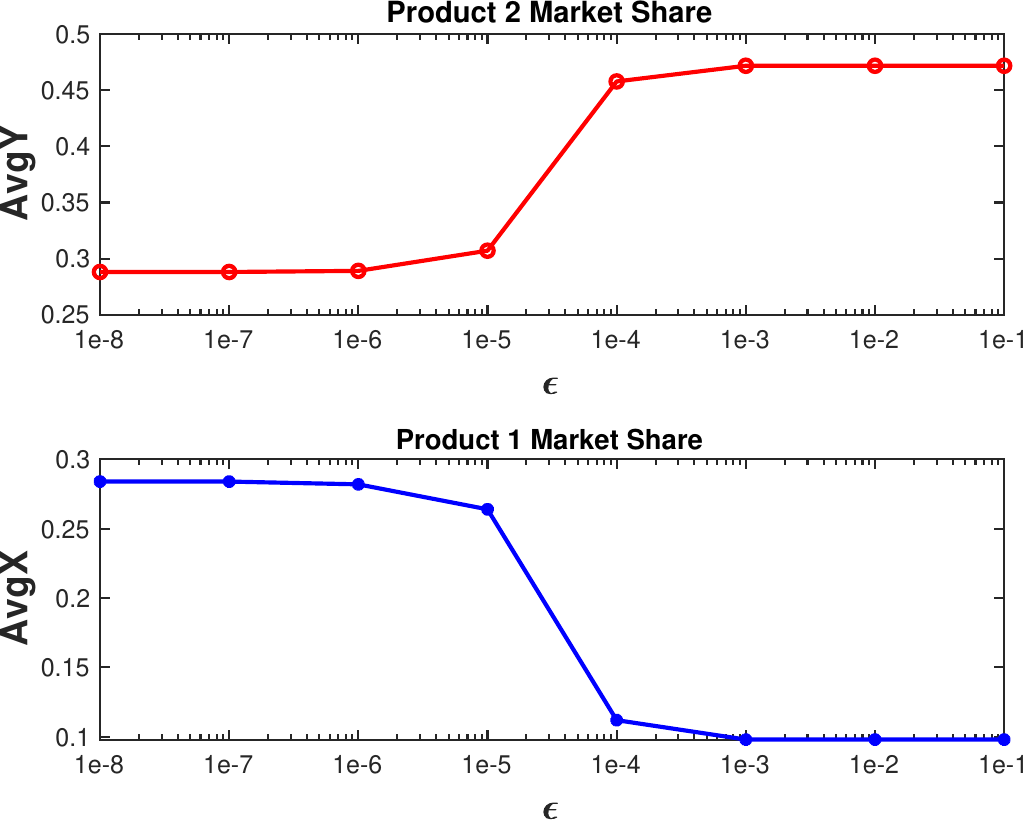}
    \caption{Facebook ($\approx12k$ nodes): Change in $\bar{\vy}$ (AvgY) and $\bar{\vx}$ (AvgX) for \textit{Homogeneously} distributed cost and different values of $\epsilon$.}
    \label{fig:diff_epsilon_fbpublichomo}
\end{figure}
%------------------------------------------
%------------------------------------------
\begin{figure}[!h]
    \centering
    \includegraphics[width = 0.78\columnwidth]{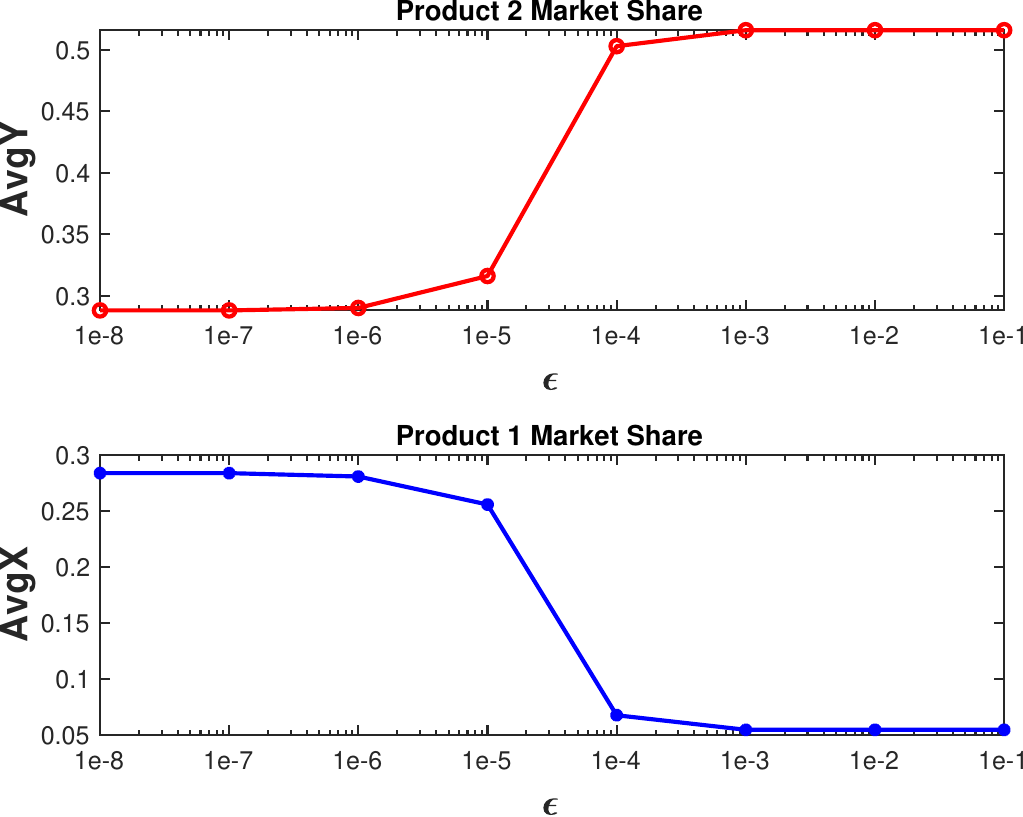}
    \caption{Facebook ($\approx12k$ nodes): Change in $\bar{\vy}$ (AvgY) and $\bar{\vx}$ (AvgX) for \textit{Heterogeneously} distributed cost and different values of $\epsilon$.}
    \label{fig:diff_epsilon_fbpublichetero}
\end{figure}
%------------------------------------------

Figure \ref{fig:baseline} shows that for the locally optimal choice of users returned by our approach for different values of $\epsilon$, under both heterogeneous and homogeneous cost distribution, the initial weaker/new product has a higher market share than the baseline centrality measures and NetShield algorithm, thus validating our theoretical results. Although it is usually assumed that spending resources on users with the highest number of followers (influential users) will result in a higher market share, our simulation results show otherwise, i.e., any set of users selected by our algorithm, irrespective of the number of followers, can result in a higher market share. This result has profound implications on the budget constraint, i.e., maximum amount that can be spent on recruiting users, as users with more followers are usually very expensive. Note that similar numerical results can indeed be obtained for even larger graphs than shown here, for any choice of budget, $C$, as long as the cost assigned to each $u_i$ is not too large and is thus, omitted here to avoid repetition. 

%------------------------------------------
\begin{figure}[!h]
    \centering
    \includegraphics[width = 0.8\columnwidth]{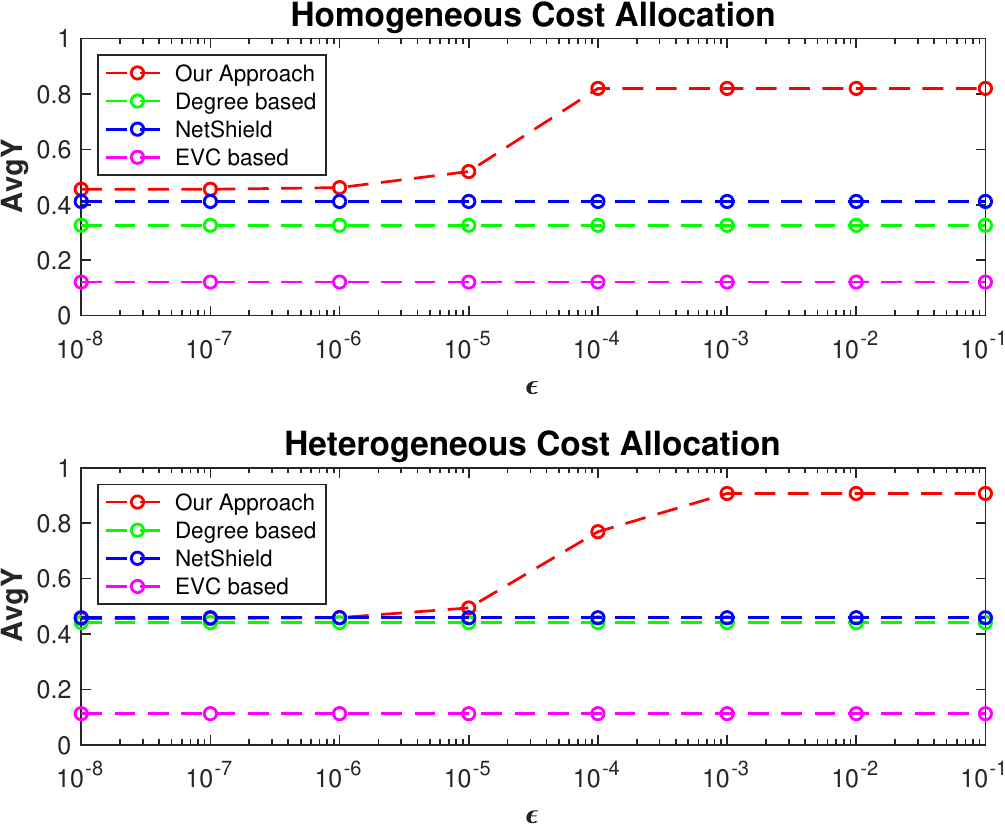}
    \caption{Facebook ($\approx4k$ nodes): Change in $\bar{\vy}$ (AvgY) for Homogeneously \& Heterogeneously distributed cost against Degree-based, Eigenvector-based centrality measure, \& NetShield Algorithm.}
    \label{fig:baseline}
\end{figure}
%------------------------------------------

%----------------------CONCLUSION--------------------------------

\section{Conclusion \& Future Work}\label{conclusion}

We have developed an algorithm that returns an optimal set of users that form a community, that is, the best possible pair of $(\gamma,\vu)$ and these set of users when targeted with advertisement under a given budget constraint results in a positive market share of the new product, Product 2. Through perturbation approach we identify the set of users that can be influenced through advertisement so that Product 2 can enter into the network with a positive market share. We also compare our proposed approach against standard centrality measures and NetShield algorithm (used in finding nodes in SIS epidemic model) using two different datasets. The results indicate that our approach of targeting the chosen users outperforms in both the cases when the budget is assigned either homogeneously or heterogeneously to users that are part of the community. We believe that our work is the first step towards finding an optimal size of community for a new firm (launching a product) to compete against a dominant product in a market. By basing our theoretical analysis using the bi-SIS epidemic model and extensive simulation on different real-world graphs demonstrate the efficacy of our approach, including its exceeding performance over the standard centrality measures.

While this paper focuses on finding an optimal set of users to maximize the market share of a new product competing against an existing dominant product, we abstract away from analyzing multiple products' (more than two) competition. As an avenue of future research, utilizing a multi-layer graph structure and incorporating multiple products’ competition would greatly contribute to a better understanding of how a new product strives to improve its market share in the presence of multiple products.

%---------------------APPENDIX--------------------------------
\appendix
\section{Appendix}\label{appendix:section3}

\subsection{Proof of Proposition \ref{prop::budget equality}}\label{appendix:prop budget equality}
%============ Proposition 4.1===========================
\begin{proof}Let $( \gamma,  \vu)$ and $(\hat \gamma, \hat \vu)$ be such that we have $ \gamma  u_i  u_j \leq \hat \gamma \hat u_i \hat u_j$, with strict inequality for at least one pair of $i,j$. Let $\vy$, $\vz$ be the resulting fixed point for Product 2 under parameters $( \gamma, \vu)$, $(\hat \gamma, \hat \vu)$), respectively. We show that $\bar{y} < \bar{z}$, that is, it is always better to spend more if the budget allows for it. Let $(\vx,\vy)$ and $(\hat \vx(t), \hat \vy(t))$ denote the fixed points of the ODE system \eqref{eq:ode_gammaUU-vector} with parameters $( \gamma, \vu)$ and $(\hat \gamma, \hat \vu)$, respectively, where the initial point $(\hat \vx(0), \hat \vy(0)) = (\vx,\vy)$. Then, we have\footnote{For any two vectors $\vp = [p_i]$ and $\vq = [q_i]$, the inequality $\vp > \vq$ implies $p_i \geq q_i$, with strict inequality for at least one entry.}
\begin{equation}
     \left.  \frac{d \hat \vx(t)}{dt}\right\vert_{t=0} = \beta_1 \text{diag}(\ones \!-\!\vx\!-\!\vy(t))\mA\vx - \delta_1 \vx = 0 \nonumber
\end{equation}

\begin{equation}
    \left. \frac{d \hat \vy(t)}{dt}\right\vert_{t=0} = \beta_2 \text{diag}(\ones\!-\!\vx\!-\!\vy)\left[ \mA+\hat\gamma \hat\vu \hat\vu^T \right]\vy -\delta_2 \vy \nonumber
\end{equation}

\begin{equation}
  \hspace{5em}  > \beta_2 \text{diag}(\ones\!-\!\vx\!-\!\vy)\left[ \mA+\gamma \vu \vu^T \right]\vy -\delta_2 \vy = 0.\nonumber
\end{equation}

Therefore, there exists an $\epsilon > 0$ for which we have $\vx \geq \hat \vx(\epsilon)$ and $\vy \leq \hat \vy(\epsilon)$. From \cite{doshi2021competing}, we know that the bi-SIS system is strongly monotone \cite{smith1995monotone} for $\vx,\vy \in (0,1)^N$. Hence, $\vy < \hat \vy(\epsilon) \ll \vy(t+\epsilon)$ for all $t > 0$, and since $\vz = \lim_{t\to \infty} \vy(t+\epsilon)$, we have $\vy \ll \vz$.
\end{proof}
%=================================================================

\subsection{Proof of Lemma \ref{lemma::initial_candidate}}
%==========Proof of Lemma 4.3===================================
\begin{proof}
Consider the matrix $\mM = \mS_{\vx^*}(\mA + \gamma \vu \vu^T)$. By substituting $\vu = \mS_{\vx^*}^{-1}\vecv$,\footnote{The vector $\vu$ can actually be taken as any scalar multiple of $\mS_{\vx^*}^{-1}\vecv$, that is $\vu = c \mS_{\vx^*}^{-1}\vecv$  for any $c \in \R$. We drop the `$c$' notation for clarity.} it can be written as 
$$\mM = \mS_{\vx^*}(\mA + \gamma \mS_{\vx^*}^{-1}\vecv\vu^T) = \mS_{\vx^*}\mA + \gamma \vecv\vu^T,$$
\noindent which is now in the form of a rank-one perturbation to the matrix $\hat{\mM} = \mS_{\vx^*}\mA$, and is in the same form as in Proposition 1.1 in \cite{bru2012brauer} with $\vecv$ as the first eigenvector of $\hat{\mM}$, and $\gamma \vu$ being the additional $N-$dimensional vector. Then, from Proposition 1.1 in \cite{bru2012brauer}, $\vecv$ is also the first eigenvector of $\mM = \mS_{\vx^*}(\mA + \gamma \vu\vu^T)$ associated with eigenvalue $\lambda(\mM) = \lambda(\mS_{\vx^*}\mA) + \gamma \vu^T\vecv = \lambda(\mS_{\vx^*}\mA) + \gamma \vu^T\mS_{\vx^*}\vu$. Finally, by rearranging the above terms, we can check that the value of $\gamma$ for which $\lambda(\mM) = 1/\tau_2$ is given by $\gamma = \frac{1/\tau_2 - \lambda(\mS_{\vx^*}\mA)}{\vu^T\mS_{\vx^*}\vu}$.
\end{proof}
%==============================================================

%=================OLD THEOREM=============================
\subsection{Proof of Theorem \ref{theorem::avg_y}}

We begin by taking the derivative of \eqref{eq:fpe y} with respect to $\gamma$ and $u_r$ evaluated at $(\gamma^c,\vu^c)$ which yields\footnote{We define \textit{c.p.} to refer to the point at which we take derivative of \eqref{eq:fpe y}, i.e., $c.p. \triangleq  (\gamma,\vu)=(\gamma^c,\vu^c)$.}
%----------------------------------------------------------------
\begin{equation}\label{eq:afterder}
    \left.\frac{\partial y_i}{\partial \gamma} \right \vert_{c.p.} = \tau_2 (1- x^*_i) \sum_{j \in \cN}( a_{ij} + \gamma^c u_i{^c} u_{j}{^c}) \left.\frac{\partial y_i}{\partial \gamma} \right \vert_{c.p.}
\end{equation}
%----------------------------------------------------------------
%----------------------------------------------------------------
\begin{equation}\label{eq:afterderu}
     \left.\frac{\partial y_i}{\partial u_r} \right \vert_{c.p.} = \tau_2 (1- x^*_i) \sum_{j \in \cN}( a_{ij} + \gamma u_i^{c} u_{j}^{c}) \left.\frac{\partial y_i}{\partial u_r} \right \vert_{c.p.},
\end{equation}
%-----------------------------------------------------------------
where $\partial \vy / \partial \gamma = [\partial y / \partial \gamma]_i$ and $\partial \vy / \partial u_r = [\partial y / \partial u_r]_i$ are eigenvectors of $\mS_{\vx^*}[\mA + \gamma \vu \vu^T]$ up to some multiplicative constants (PF theorem \cite{meyer2000matrix}) corresponding to eigenvalue $\tau_2$. We first give the following result, to help us derive Lemma \ref{lemma::avg_y-old}, which in turn helps in determining the multiplicative constants.

%===========Proof of 1st lemma==================================
\begin{lemma}\label{lemma::eig_val_perturbation}
Let $\epsilon_1, \epsilon_2 > 0$ be any small constants such that $\gamma(\epsilon_1)$ and $\vu(\epsilon_2)$ are perturbations around $\gamma, \vu$. Then, for all $\epsilon_3 \!>\! 0$ such that $\lambda(\gamma,\vu) = \lambda(\gamma^c,\vu^c) + \epsilon_3$, there exists $\epsilon_4 \!>\!0$ such that $\vy = \epsilon_4 \vecv(\gamma,\vu)$.
\end{lemma}
\begin{proof} 
Since \eqref{eq:afterder} is an eigenvalue problem of the form $\mathbf{A} \nu = \lambda \mathbf{\nu}$, it allows for $\left.\frac{d y_{i}}{d \gamma}\right \vert_{\gamma = \gamma^c} \!>\! 0$ as $\left.\lambda(\mS_{\vx^*}[\mA+\gamma\vu\vu^T])\right \vert_{\gamma = \gamma^{c}} = 1/\tau_2$. Similarly, it can be shown that $\left.\frac{d y_{i}}{d u_i}\right \vert_{u_i = u_i^c} \!>\! 0$.

From \eqref{eq:afterder}, we have that $1/\tau_2 = \lambda(\gamma^c, \vu^c) > 0$ and $\frac{d y_i}{d \gamma} \triangleq [v]_i$, i.e., $\vy = \epsilon \vecv(\gamma,\vu)$ is the eigenvector of $\mS_{\vx^*}[\mA+\gamma\vu\vu^T]$. For $\lambda(\gamma^c, \vu^c) > \lambda(\gamma, \vu)$, $1/\tau_2$ cannot be the eigenvalue and only possible solution is $[v]_i = 0$. Hence, for $\lambda(\gamma, \vu) > \lambda(\gamma^c, \vu^c)$, i.e., for any small arbitrary constant $\epsilon,~\lambda(\gamma, \vu)\!=\!\lambda(\gamma^c, \vu^c) + \epsilon$,  $1/\tau_2$ is the eigenvalue with $\vy = \epsilon \vecv(\gamma,\vu)$.
\end{proof}
% %================================================================
Now, the following lemma helps in determining the constants.

%=============== 2nd Lemma  ====================================
\begin{lemma}\label{lemma::avg_y-old}
    Let $\lambda(\gamma, \vu) \triangleq \lambda(\mS_{\vx^*}[\mA + \gamma \vu \vu^T])$, and $(\gamma^c,\vu^c)$ be any pair of critical parameters, and let $\bar \vy(\gamma,\vu) \triangleq \ones^T\vy(\gamma\vu)/N$ denote the market share as a function of any feasible $(\gamma,\vu)$. Then, $\bar \vy(\gamma,\vu)$ can be written as,
  \begin{equation}\label{market share given lambda}
  \begin{aligned}
     \bar{\vy}(\gamma,\vu) &= \frac{\tau_2\ones^T \vecv(\gamma^c,\vu^c)}{N \sum \limits_{i=1}^{N} \frac{v_{i}(\gamma^c,u_i^c)^3}{(1-x_i^*)^2}} \left( \lambda(\gamma,\vu) - 1/\tau_2 \right) \\
     &+ 
     O\left( \left( \lambda(\gamma,\vu) - 1/\tau_2 \right)^2 \right).
   \end{aligned}
  \end{equation}
\end{lemma}
%==================================================================
\begin{proof}

From \eqref{eq:fpe x} and \eqref{eq:fpe y}, we have the following expressions and then Proposition \ref{similarity trans} regarding orthogonality of right eigenvectors.
\begin{equation}\label{sum_y}
    \frac{\ones^T\vy}{\tau_2} = (\ones-\vx)^T\!\left[\mA \!+\! \gamma \vu\vu^T \right]\!\vy - \vy^T\!\left[\mA \!+\! \gamma \vu\vu^T \right]\!\vy    
\end{equation}

\begin{equation}\label{quad_y}
\begin{aligned}
    \frac{\vy^T\mS_{\vx}^{-1}\vy}{\tau_2} &= \vy^T\!\left[\mA \!+\! \gamma \vu\vu^T \right]\!\vy \\
    & - \vy^T\mS_{\vx}^{-1}\text{diag}(\vy)\!\left[\mA \!+\! \gamma \vu\vu^T \right]\!\vy
\end{aligned}
\end{equation}

\begin{proposition}\label{similarity trans}
For all $k \neq j$, the right eigenvectors $\vecv_k$ of $\mM = \mS_{\vx^*}[A + \gamma \vu \vu^T]$ satisfy $\vecv_k^T\mS_{\vx}^{-1}\vecv_j = 0$.
\end{proposition}
\begin{proof}
Let, $\mS_{\vx}^{-1/2}\mM\mS_{\vx^*}^{1/2} = \mS_{\vx^*}^{1/2}[\mA + \gamma\vu\vu^T]\mS_{\vx^*}^{1/2} \triangleq \hat \mM$ hence, $\mM$ and $\hat \mM$ are similar matrices with same set of eigenvalues \cite{meyer2000matrix}, with right eigenvector of $\hat \mM$ corresponding to eigenvalue $\lambda_k$ given by $\mS_{\vx^*}^{-1/2}\vecv_k$. $\hat \mM$ is a symmetric matrix as $\mA + \gamma\vu\vu^T$ is symmetric, and its eigenvectors are orthogonal to each other. Hence, for all $k \neq j$, we have $\vecv_k^T \mS_{\vx^*}^{-1/2} \mS_{\vx^*}^{1/2} \vecv_j = \vecv_k^T \mS_{\vx^*}^{-1} \vecv_j = 0$.
\end{proof}
% %------------------------------------------------------------
\noindent From $\vy$ as a linear combination of eigenvectors, we have
%----------------------------------------------------
\begin{equation}\label{linear_combination}
    \vy = \sum\limits_{k=1}^N c_k \vecv_k.
\end{equation}
%----------------------------------------------------
From the above equation, we have

%----------------------------------------------------
\vspace{-3.4em}
\begin{multicols}{2}
  \begin{equation}
    \frac{\ones^T\vy}{\tau_2} = \sum\limits_{k=1}^N \tau_2^{-1}c_k \ones\vecv_k
    \label{sum_y_vk}, \\
  \end{equation}\break
  \begin{equation}
    \frac{\vy^T\mS_{\vx}^{-1}\vy}{\tau_2} = \sum\limits_{k=1}^N \tau_2^{-1}c_k^2.
    \label{quad_u_vk}
  \end{equation}
\end{multicols}
%----------------------------------------------------
\noindent Substituting \eqref{linear_combination} in \eqref{sum_y} and in \eqref{quad_y} gives
%----------------------------------------------------
\begin{equation}
    \begin{aligned}[b]\label{sum_y_full}
         \frac{\ones^T\vy}{\tau_2}
        &= \sum\limits_{k=1}^N c_k\lambda_k  \ones^T \vecv_k - c_k^2 \lambda_k.
    \end{aligned}
\end{equation}
%----------------------------------------------------
\begin{equation}
    \begin{aligned}[b]\label{quad_y_full}
        \begin{split} 
        \frac{\vy^T\mS_{\vx}^{-1}\vy}{\tau_2}
        &= \sum\limits_{j=1}^N c_j^2 \lambda_j \\
        & - \sum\limits_{j=1}^N\sum\limits_{k=1}^N\sum\limits_{l=1}^N c_j c_k c_l\lambda_l \sum\limits_{i=1}^N \frac{[\vecv_j]_i[\vecv_k]_i[\vecv_k]_i}{(1-x_i)^2} 
     \end{split}
    \end{aligned}
\end{equation}
%----------------------------------------------------
%----------------------------------------------------
By equating \eqref{sum_y_vk} with \eqref{sum_y_full}, \eqref{quad_u_vk} with \eqref{quad_y_full} yields
%----------------------------------------------------
\begin{equation}
    \sum\limits_{k=1}^{N} (\lambda_k - \tau_2^{-1})c_k\ones^T\vecv_k = \sum\limits_{k=1}^{N}c_k^2 \lambda_k
    \label{first}
\end{equation}
\begin{equation}
\small
    \sum\limits_{j=1}^{N}(\lambda_j - \tau_2^{-1})c_j^2 = \sum\limits_{j=1}^N\sum\limits_{k=1}^N\sum\limits_{l=1}^N c_j c_k c_l\lambda_l \sum\limits_{i=1}^N \frac{[\vecv_j]_i[\vecv_k]_i[\vecv_k]_i}{(1-x_i)^2} 
    \label{second}   
\end{equation}
%----------------------------------------------------

From Lemma \ref{lemma::eig_val_perturbation}, we know $\lambda(\gamma,\vu) = \lambda(\gamma^c,\vu^c) + \epsilon_1$ for some $\epsilon_1>0$, there exists an $\epsilon_2>0$ such that $\vy = \epsilon_2 \vecv(\gamma^c,\vu^c)$. Also, $\vy = \alpha \vecv(\gamma^c,\vu^c) + \beta\vw$, where $\vw$ is a vector orthogonal to $\vecv(\gamma,\vu)$, and $\beta$ goes down to zero faster than $\alpha$ as $(\gamma,\vu)\to(\gamma^c,\vu^c)$, or $\lambda(\gamma,\vu)\to\lambda(\gamma^c,\vu^c) \tau_2^{-1}$. Let $\alpha = \alpha_0(\lambda_1 - \lambda_1^c)^p + o((\lambda_1 - \lambda_1^c)^p)$ and $\beta = \beta_0(\lambda_1 - \lambda_1^c)^q +o((\lambda_1 - \lambda_1^c)^q)$, where $q>p>0$. From \eqref{linear_combination}, $\vy$'s component $c_1$ and $c_k$ (for all $k>1$) is associated with $\vecv$ and vectors orthogonal to $\vecv$, respectively. Then, we have
$c_1 = r_1(\lambda_1 - \lambda_1^c)^p + o((\lambda_1 - \lambda_1^c)^p)$, and $c_k = r_k(\lambda_1 - \lambda_1^c)^q + o((\lambda_1 - \lambda_1^c)^q)$. Substituting these in \eqref{first} when $\lambda_1 \to \lambda_1^c$ shows that it is of the order, \\
$\sum\limits_{k=2}^{N} r_k\ones^T\vecv_k(\lambda_k - \tau_2^{-1})(\lambda_1 - \lambda_1^c)^q + r_1\ones^T\vecv_1(\lambda_1 - \lambda_1^c)^{p+1} \\
+ o\left((\lambda_1 -\lambda_1^c)^{\min(q,p+1)}\right)$, and 
$\sum\limits_{k=2}^{N}r_k^2 \lambda_k(\lambda_1 - \lambda_1^c)^{2q} + r_1^2\lambda_1(\lambda_1 - \lambda_1^c)^{2p} +  o\left((\lambda_1 - \lambda_1^c)^{2p}\right)$.

Similarly \eqref{second} gives,
\begin{equation*}
    \begin{aligned}[b]
        &r_1^2(\lambda_1 - \lambda_1^c)^{2p+1}
        + \sum\limits_{k=2}^{N}r_k^2 (\lambda_k - \tau_2^{-1})(\lambda_1 - \lambda_1^c)^{2q} \\
        +& o\left((\lambda_1 - \lambda_1^c)^{\min(2q,2p+1)}\right),
    \end{aligned}
\end{equation*}
\begin{equation*}
    \begin{aligned}[b]
        &r_1^3\lambda_1(\lambda_1 - \lambda_1^c)^{3p}\sum\limits_{i=1}^N \frac{v_{1i}}{(1-x_i)^2} 
        + o((\lambda_1 - \lambda_1^c)^{3p}) \\
        &+ O((\lambda_1 - \lambda_1^c)^{2p+q}).
    \end{aligned}
\end{equation*}

From these we have the following cases: \textit{Case A:} $q\geq p+1$, which means that $p+q=2p$ implying $p=1$; \textit{Case B:} $q<p+1$, implying $q=2p$; \textit{Case C:} $2q\geq 2p+1$, in which case $2p+1 = 3p$ and $2q = 2p+q$, implying $p=1$ and $q=2$; \textit{Case D:} $2q<2p+1$, which gives us $2q = 3p$ and $2p+1 = 2p+q$. Observe that Cases A and C are the only ones complementing each other, while the others pairs are contradictory. Hence, it's true that $p=1$ and $q=2$. By equating the corresponding powers allows us to rewrite \eqref{linear_combination} as
\begin{align*}
    \vy = \frac{\vecv(\gamma,\vu)}{\lambda_1(\gamma,\vu)\sum\limits_{i=1}^N \frac{v_i(\gamma,\vu)}{(1-x_i(\gamma,\vu))^2}}\left(\lambda_1(\gamma,\vu) - \lambda_1(\gamma^c,\vu^c)\right) \\
    + O\left((\lambda_1(\gamma,\vu) - \lambda_1(\gamma^c,\vu^c))^2\right).
\end{align*}
A Taylor series expansion with respect to $\lambda_1(\gamma,\vu)$ and centred around $\lambda_1(\gamma^c,\vu^c)=\tau_2^{-1}$, gives us\footnote{$\lambda_1(\gamma,\vu)$ corresponds to the first eigenvalue from PF Theorem and is equivalent to $\lambda(\gamma,\vu)$.}
\begin{align*}
    \begin{split}
        \vy & = \frac{\tau_2\vecv(\gamma^c,\vu^c)}{\sum\limits_{i=1}^N \frac{v_i(\gamma^c,\vu^c)}{(1-x^*_i)^2}}\left(\lambda_1(\gamma,\vu) - \tau_2^{-1}\right) \\ 
        & + O\left((\lambda_1(\gamma,\vu) - \tau_2^{-1})^2\right).
    \end{split}
\end{align*}
\end{proof}

%=====Proof of Theorem4.2======================================
Now, we use the results from Lemma \ref{lemma::avg_y-old} towards the proof of Theorem \ref{theorem::avg_y}, which is as follows. Note that $\lambda(\gamma,\vu) - 1/\tau_2$ on the RHS in Lemma \ref{lemma::avg_y-old} is actually based on given $(\gamma^c,\vu^c)$ at the critical point, i.e., $\vx=\vx^*, \vy=\boldsymbol{0}$.

\begin{proof}
From first-order approximation, we know that $f(x) \approx f(a) + f'(a) (x-a)$. Utilizing this, $\lambda(\gamma^c, \vu)$ with respect to change in $u_i$ based on given $\lambda(\gamma^c, \vu^c)$ at critical point, we have
\begin{equation}\label{first-order approx}
    \lambda(\gamma, \vu) \approx \lambda(\gamma^c, \vu^c) + \lambda'(\gamma^c, \vu^c) \Delta \vu.
\end{equation}
Utilizing Theorem 1 from \cite{greenbaum2020first}, we know that at $u_i = u_i^c$, $\lambda'(\gamma^c, \vu^c) = \vq^* \mM' \vp$ where $\lambda'(\gamma^c, \vu^c)$ and $\mM'$ are, respectively, the derivatives of $\lambda(\gamma, \vu)$ and $\mM$ at $u_i = u_i^c$. Also, $\mM = \mS_{\vx^*}(\mA + \gamma \vu\vu^T)$ with largest eigenvalue $\lambda(\gamma^c, \vu^c)$ corresponding to right eigenvector $\vp = \vecv$ and left eigenvector $\vq = \vecv ^T \mS_{\vx^*}^{-1}$. Substituting this in \eqref{first-order approx} and from Lemma \ref{lemma::initial_candidate}, we have
\begin{equation}\label{derivative}
    \lambda(\gamma^c, \vu) -  \frac{1}{\tau_2} = \vq^* \mM' \vp \Delta \vu,
\end{equation}
where $\Delta \vu = u_i - u_i^c$. The derivative of $\mM$ with respect to $u_i$ gives $\frac{\partial \mM}{\partial u_i} = \gamma^c \mS_{\vx^*}\big[e_i \vu^T + \vu e_i^T\big]$.

This gives us $\mM' = \gamma \mS_{\vx^*}\big[e_i \vu^T + \vu e_i^T\big]$ at $u_i$. Substituting this on the RHS of~\eqref{derivative} gives $\vq^* \mM' \vp \Delta \vu = 2\gamma^c \vu^T \vecv(u_i - u^c_i)v_i$, where the term $2 \gamma^c \vu^T$ always remains constant. Substituting this again in \eqref{derivative} gives
\begin{equation}\label{eq:lambdatilde}
    \lambda(\gamma^c, \vu) -  \frac{1}{\tau_2} = 2\gamma^c \vu^T \vecv(u_i - u^c_i)v_i,
\end{equation}
which implies that the change is in the direction of $\vecv =[v_i]$ for all $i \in \cN$.
\end{proof}
%=============================================================
\section*{Broader Impact}
Our goal is to develop a heuristic approach through which a new/weaker product can survive (with a positive market share) along with an existing dominant product. This has wider managerial implications for brand managers running new product introduction campaigns. Specifically, such campaigns involve targeting niche communities through advertisements or promotional offers. For example, brand managers utilize social media platforms such as Instagram, where they perform brand advertisements through influencers/niche groups, while operating within a limited advertisement budget. Our model captures this exact dynamics by identifying an optimal set of users under a budget constraint to form a community who can be similarly targeted. 

We use real-world social network data of Facebook users and compare it against standard centrality measures. The comparison results indicate that it is not always necessary to spend all resources on more famous influencers/public figures for product promotion, but the community identified by our approach provides a much better market share. Again, this has wider implications and provides insights on how best to spend the limited advertisement budget. Finally, we assume that companies have the complete network information, which might not be always true. However, even with incomplete network information and through our approach, companies can identify niche set of users and ensure survival. Further, more information regarding the network will only result in improving the market share.

\section*{Ethical Considerations}
The data (graph structure) for simulation results is accessed from publicly available repositories. There is no other ethical issue associated with our work.

\section*{Acknowledgements}
This work was done when Shailaja Mallick and Vishwaraj Doshi were PhD students at North Carolina State University. This work was supported in part by National Science Foundation under Grant Nos. CNS-2007423 and IIS-1910749.

{\fontsize{9.8pt}{10.8pt} \selectfont
\bibliography{uv}}

\end{document}